\documentclass[sn-mathphys-num]{sn-jnl}

\usepackage{graphicx}%
\usepackage{multirow}%
\usepackage{amsmath,amssymb,amsfonts}%
\usepackage{amsthm}%
\usepackage{mathrsfs}%
\usepackage[title]{appendix}%
\usepackage{xcolor}%
\usepackage{textcomp}%
\usepackage{manyfoot}%
\usepackage{booktabs}%
\usepackage{algorithmicx}%
\usepackage{algpseudocode}%
\usepackage{listings}%
\usepackage{epstopdf}
\usepackage{fullpage}
\usepackage{setspace}
\usepackage{lscape}
\usepackage{amsthm}
\usepackage{acronym}
\usepackage{mathtools}
\usepackage{booktabs}
\usepackage{hyperref}
\usepackage{setspace}
\usepackage{natbib}
\usepackage{color}
\usepackage{caption}
\usepackage{amsfonts}
\usepackage{bm}
\usepackage{mathrsfs}
\usepackage{mathptmx}
\usepackage{mathrsfs}
\usepackage{subcaption}
\usepackage{latexsym}
\usepackage{nicefrac}
\usepackage{verbatim}
\usepackage{lineno}
\theoremstyle{thmstyleone}%
\newtheorem{Theorem}{Theorem}
\theoremstyle{thmstyletwo}%
\newtheorem{Lemma}{Lemma}
\theoremstyle{thmstylethree}%

\begin{document}

\title[Article Title]{Shrinkage estimators in zero-inflated Bell regression model with application}


\author*[1]{\fnm{Solmaz} \sur{Seifollahi}}\email{s.seifollahi@tabrizu.ac.ir}

\author[2]{\fnm{Hossein} \sur{Bevrani}}\email{bevrani@gmail.com}
\equalcont{These authors contributed equally to this work.}

\author[3]{\fnm{Zakariya} \sur{Yahya Algamal}}\email{zakariya.algamal@uomosul.edu.iq}
\equalcont{These authors contributed equally to this work.}

\affil*[1]{\orgdiv{Department of Statistics}, \orgname{University of Tabriz}, \city{Tabriz}, \country{Iran}}

\affil[2]{\orgdiv{Department of Statistics}, \orgname{University of Tabriz}, \city{Tabriz},  \country{Iran}}

\affil[3]{\orgdiv{College of Computer Science and Mathematics}, \orgname{University of Mosul}, \city{Mosul}, \country{Iraq}}


\abstract{We propose Stein-type estimators for zero-inflated Bell regression models by incorporating information on model parameters. These estimators combine the advantages of unrestricted and restricted estimators. We derive the asymptotic distributional properties, including bias and mean squared error, for the proposed shrinkage estimators. Monte Carlo simulations demonstrate the superior performance of our shrinkage estimators across various scenarios. Furthermore, we apply the proposed estimators to analyze a real dataset, showcasing their practical utility.}


\maketitle

\section{Introduction}
The introduction of generalized linear models (GLMs) by Nelder and Wedderburn \cite{Nelder} extended the application of linear models to handle data that deviate from the assumptions of a normal distribution, such as binary or count data. This breakthrough marked a major milestone in the field of statistics.
In the context of modeling count data, the Poisson regression model is the most commonly used. However, this model is practical only when the mean of the response variable is equal to its variance. Real data often demonstrates that the variance is larger than the mean, a phenomenon known as overdispersion. As a response to this issue, the Bell regression model, proposed by Castellares et al. \cite{Castellares}, offers an alternative approach to handle overdispersion. The majority of studies conducted using the Bell regression model focus on estimating model coefficients in the presence of multicollinearity issues. For instance, the  For example, Ridge estimator \cite{Amin}, Liu estimator \cite{Majid}, Liu-type estimator \cite{Ertan}, Jackknifed ridge estimator \cite{Algamal1} and  Jackknifed L-K estimator \cite{Abduljabbar} have been studied in this context.

When using Bell regression models, another issue that may arise is the frequency of zeros being more than usual. In such situations, the zero-inflated Bell (ZIBell) regression model, introduced in \cite{Lemonte} and later studied in \cite{Ali}, can be used. In the presence of multicollinearity, the Liu and ridge estimators have been considered for ZIBell regression in \cite{Algamal2}.

Another approach to address multicollinearity is to eliminate the inactive or non-significant covariates from the model. This can be achieved by imposing restrictions on the model coefficients, which results in a new estimator referred to as the restricted estimator. To enhance the accuracy of estimation based on restrictions on model coefficients, researchers often use shrinkage methods. These methods have been discussed in various regression models, such as gamma \cite{Akram}, beta \cite{Seifollahi1}, Poisson \cite{Amin}, negative binomial \cite{Zandi}, inverse Gaussian \cite{Akram2} and lately Bell \cite{Seifollahi2} regression models.
For more detail about shrinkage methods one can refer to \cite{Ahmed}. However, there is a lack of studies that consider shrinkage methods in ZIBell models.
The main goal of this paper is to devise shrinkage techniques, such as Stein-type and preliminary estimators, for the ZIBell models. These methods are developed based on specific information about the model coefficients.

The structure of the paper is as follows: The ZIBell regression model and the maximum likelihood estimator of coefficients are described in section \ref{sec2}. The proposed estimators based on the information about coefficients are proposed in Section \ref{sec3}. The statistical features of the proposed estimators are discussed in Section \ref{sec4}.  Section \ref{sec5} concludes a comparison of the proposed estimators performance based on a Mont Carlo simulation by considering various scenarios. Finally, in Section \ref{sec6}, we analyze the empirical data, and the conclusion is provided in Section \ref{sec7}, along with a concluding remark.

\section{Zero-inflated Bell regression model}\label{sec2}
To model count data with excess zeros and overdispersion, Lemonte et al. \cite{Lemonte} recommended applying ZIBell regression.  Suppose that $\boldsymbol{y}=(y_1, y_2, \ldots, y_n)^T$ denotes the vector of $n$ count measurements of a response variable. If $y_i$ follows a ZIBell distribution with parameter $(\mu_i, p_i)$, its probability mass function is given by:
\begin{equation}
f(y_i; \mu_i)= \begin{cases}
                         p_i + (1-p_i) exp(1-e^{W_0(\mu_i)})& y_i=0 \\[5pt]
                         (1-p_i)exp(1-e^{W_0(\mu_i)}) \dfrac{W_0(\mu_i)^{y_i}B_{y_i}}{{y_i}!} & y_i >0
                    \end{cases}
\end{equation}
where $p_i \in [0,1]$ and $B_y=\dfrac{1}{e} \sum_{l=0}^\infty l^n/l!$ denotes the Bell numbers \cite{Bell1, Bell2}. The mean and variance of ZIBell distribution are $E(Y_i)= \mu_i (1-p_i)$  and $Var(Y_i)= \mu_i (1-p_i) [1+W_0(\mu_i)+\mu_i p_i]$. To define the ZIBell regression model, two following link functions can be used:
\begin{align}
\log (\mu_i)&= \eta_{i1}=\boldsymbol{x}_i^T\boldsymbol{\beta} \label{mui}\\
log(\dfrac{p_i}{1-p_i})&=\eta_{i2}= \boldsymbol{z}_i^T\boldsymbol{\gamma}\label{pii}
\end{align}
Here, $\boldsymbol{x}_i=(x_{i1}, x_{i2}, \ldots, x_{ip})$ and $\boldsymbol{z}_i=(z_{i1}, z_{i2}, \ldots, z_{iq})$ denotes the vector of covariates associated with response $y_i$. If $x_{i1}=z_{i1}=1$, the regression models will include the intercept terms. $\boldsymbol{\beta}= (\beta_1, \beta_2, \ldots, \beta_p)^T$ is the coefficients related to the count part of the model and $\boldsymbol{\gamma}= (\gamma_1, \gamma_2, \ldots, \gamma_q)^T$ is the coefficients related to the zero part. We assume that $\boldsymbol{\theta}=(\boldsymbol{\beta}^T,\boldsymbol{\gamma}^T)^T$ is a $k \times 1$ vector of all model coefficients where $k=p+q$.

Based on random sample $\boldsymbol{y}=(y_1, y_2, \ldots, y_n)^T$, the log-likelihood function of the ZIBell model is given by:
\begin{equation}\label{LF}
\ell(\boldsymbol{\theta}; \boldsymbol{y})= \sum_{\{i; y_i=0\}} \log \big[ e^{\eta_{i2}}+e^{1-e^{W_0(\mu_i)}}\big]- \sum_{\{i; y_i>0\}} \big[y_i \log\left(W_0(\mu_i)\right) -e^{W_0(\mu_i)}\big]- \sum_{i=1}^{n} \log\left(1- e^{\eta_{i2}}\right)
\end{equation}

The MLE or unrestricted ZIBell (UNZIB) estimator of $\boldsymbol{\theta}=(\boldsymbol{\beta}^T,\boldsymbol{\gamma}^T)^T$ is obtained by finding the roots of the score function which is defined as $S(\boldsymbol{\theta})= \dfrac{\partial}{\partial \boldsymbol{\theta}^T}\ell(\boldsymbol{\theta}; \boldsymbol{y})$. It is obvious from \eqref{LF} that $S(\boldsymbol{\theta})$ contains some non-linear system equations with no close-form. Hence, The UNZIB estimator need to be obtained through a numerical optimization algorithm such as the Newton-Raphson iterative technique. See \cite{Lemonte} for more details.
If we denote the UNZIB estimator by $\hat{\boldsymbol{\theta}}_{UNZIB}$, under regularity conditions,  $\hat{\boldsymbol{\theta}}_{UNZIB}$ are asymptotically normal, asymptotically unbiased with the covariance matrix given by inverse of Fisher information matrix.  Suppose that $\boldsymbol{F}$ is the $(k\times k)$ Fisher information matrix for $\boldsymbol{\theta}$. Thus, when $n \rightarrow \infty$, we have
\begin{equation}\label{distun}
\sqrt{n}(\hat{\boldsymbol{\theta}}_{UNZIB} - \boldsymbol{\theta})\sim N_k(\boldsymbol{0}, \boldsymbol{F}^{-1}(\boldsymbol{\theta}))
\end{equation}
and
\begin{equation*}
\boldsymbol{F}= \bigg[\begin{matrix} X^TV_1X& X^TV_2Z\\ Z^TV_2X& Z^TV_3Z\end{matrix}\bigg]
\end{equation*}
where $V_1= diag(v_{11}, \ldots, v_{1n})$, $V_2= diag(v_{21}, \ldots, v_{2n})$, $V_3= diag(v_{31}, \ldots, v_{3n})$ and

\begin{align*}
v_{1i}&= v_{1i}^{(0)} I_{\{y_i=0\}}+v_{1i}^{(1)} I_{\{y_i>0\}}\\[5pt]
v_{1i}^{(0)} &=-\dfrac{exp(1+2 \eta_{i1}+\eta_{i2}- e^{W_0(\mu_i)})}{[e^{\eta_{i2}}+exp(1-e^{W_0(\mu_i)})]^2 [1+W_0(\mu_i)]^2}\\[5pt]
& \qquad -\dfrac{exp(1+ \eta_{i1}- e^{W_0(\mu_i)})}{[e^{\eta_{i2}}+exp(1-e^{W_0(\mu_i)})] [1+W_0(\mu_i)]}\bigg\{1- \dfrac{W_0(\mu_i)}{[1+W_0(\mu_i)]^2} \bigg\}\\[5pt]
v_{1i}^{(1)} &= -\dfrac{\mu_i}{1+W_0(\mu_i)}-\dfrac{W_0(\mu_i)}{[1+W_0(\mu_i)]^3}(y_i-\mu_i)\\[5pt]
v_{2i} &= \dfrac{exp(1+ \eta_{i1}+\eta_{i2}- e^{W_0(\mu_i)})}{[e^{\eta_{i2}}-exp(1-e^{W_0(\mu_i)})]^2 [1+W_0(\mu_i)]}I_{\{y_i=0\}}\\[5pt]
v_{3i} &= \dfrac{exp(1+\eta_{i2}- e^{W_0(\mu_i)})I_{\{y_i=0\}}}{[e^{\eta_{i2}}+exp(1-e^{W_0(\mu_i)})]^2 }- \dfrac{e^{\eta_{i2}}}{(1+e^{\eta_{i2}})^2}    \end{align*}

\section{Shrinkage Methods} \label{sec3}
In this section, we introduce shrinkage estimators based on two models: the full model, which includes all model coefficients, and the sub-model, which contains only active model coefficients. To include only significant coefficients in the model, we need to omit the non-significant coefficients by restricting as follows:
\begin{equation}\label{H0}
 \boldsymbol{R}\boldsymbol{\theta}=\boldsymbol{r}
\end{equation}
where $\boldsymbol{R}$ and $\boldsymbol{r}$ are respectively a pre-specified $p_2\times k $ matrix and $p_2\times 1 $ vector based on non-significant coefficients in the model and $p_2<k$ is the number of non-significant coefficients.
For the full model, the model coefficients will be estimated by the maximum likelihood method as explained in previous section. However, for sub-model, the model coefficients estimation can be obtained by maximizing log-likelihood function in \eqref{LF} under restrictions defined in \eqref{H0}.
The restricted ZIBell (REZIB) estimator is obtained based on this approach, according to  \cite{Heyde} as follows:
\begin{equation}\label{RES}
\hat{\boldsymbol{\theta}}_{REZIB}= \hat{\boldsymbol{\theta}}_{UNZIB} - \boldsymbol{F}^{-1}  \boldsymbol{R}^T \big[\boldsymbol{R} \boldsymbol{F}^{-1}\boldsymbol{R}^T\big]^{-1} (\boldsymbol{R}\hat{\boldsymbol{\theta}}_{UNZIB}- \boldsymbol{r})
\end{equation}
The REZIB estimator is entirely dependent on the validity of the restriction in equation \eqref{H0}. Before using the REZIB estimator, it is crucial to verify this restriction. The likelihood ratio test statistics for testing $H_0 :~ \boldsymbol{R}\boldsymbol{\theta}=\boldsymbol{r}$ vs  $H_1 :~ \boldsymbol{R}\boldsymbol{\theta} \neq \boldsymbol{r}$ is given by:
\begin{align}\label{TestSt}
T_n &= 2 \big[\ell(\hat{\boldsymbol{\theta}}_{UNZIB})- \ell(\hat{\boldsymbol{\theta}}_{REZIB})\big]\nonumber \\
&= (\boldsymbol{R}\hat{\boldsymbol{\theta}}_{UNZIB} - \boldsymbol{h})^T \big[\boldsymbol{R} \boldsymbol{F}^{-1}\boldsymbol{R}^T\big]^{-1}(\boldsymbol{R}\hat{\boldsymbol{\theta}}_{UNZIB} - \boldsymbol{r})
\end{align}
where $\ell(\hat{\boldsymbol{\theta}}_{UNZIB})$ and $\ell(\hat{\boldsymbol{\theta}}_{REZIB})$ are the log-likelihood function evaluated at $\hat{\boldsymbol{\theta}}_{UNZIB}$ and $\hat{\boldsymbol{\theta}}_{REZIB}$, respectively. Under $H_0$, for large enough  $n$, $T_n$ asymptotically follows a chi-squared distribution with $p_2$ degrees of freedom.

\subsection{Stein-type method}
Shrinkage methods aim to optimally combine unrestricted and restricted estimators to enhance the performance of the unrestricted estimator. Two well-known methods are the James-Stein and Positive James-Stein methods, which shrink the coefficients towards the center of the parameter space, resulting in improved predictions compared to ordinary least squares or maximum likelihood methods. These methods yield estimators that outperform the maximum likelihood estimator by having a lower or equal mean squared error. This implies that these estimators can provide more accurate estimates, especially when the number of coefficients is large. This subsection provides an explanation of these methods
The James-Stein ZIBell (JSZIB) estimator is given by:
\begin{equation}\label{SE}
\hat{\boldsymbol{\theta}}_{JSZIB}= \hat{\boldsymbol{\theta}}_{REZIB}+ \bigg(1- \dfrac{p_2-2}{T_n}\bigg)
(\hat{\boldsymbol{\theta}}_{REZIB}-\hat{\boldsymbol{\theta}}_{UNZIB}) \qquad p_2 \geq 3
\end{equation}
When using the JSZIB estimator, it is possible to encounter a situation where $T_n < p_2-2$, leading to over-shrinkage during the estimation process. To address this limitation and enhance the estimator's performance in estimation tasks, the Positive James-Stein ZIBell (PJSZIB) estimator is developed. This new estimator aims to mitigate the over-shrinkage issue associated with the JSZIB estimator, thereby improving its overall performance in estimation.
The PJSZIB estimator is given by:
\begin{equation}\label{PSE}
\hat{\boldsymbol{\theta}}_{JSZIB}= \hat{\boldsymbol{\theta}}_{REZIB}+ \bigg(1- \dfrac{p_2-2}{T_n}\bigg)^{+}
(\hat{\boldsymbol{\theta}}_{REZIB}-\hat{\boldsymbol{\theta}}_{UNZIB}) \qquad p_2 \geq 3
\end{equation}
where $z^+= max \{0, z\}$ \cite{Kibria}.

\subsection{Preliminary Test Method}
The preliminary test method in regression models involves using preliminary test estimators to enhance the accuracy of coefficient estimation. These estimators are employed when there are suspicions of model misspecification or when certain linear restrictions on the regression coefficients are suspected. The preliminary test  estimator is designed to address these concerns and improve the estimation process.
The Preliminary test ZIBell (PTZIB) estimator has the following form:
\begin{equation}\label{PTE}
\hat{\boldsymbol{\theta}}_{PTZIB}= \hat{\boldsymbol{\theta}}_{UNZIB}+ (\hat{\boldsymbol{\theta}}_{REZIB}-\hat{\boldsymbol{\theta}}_{UNZIB})I_{(T_n< T_{p_2, \alpha})}
\end{equation}
where $I(.)$ is the indicator function and $\alpha$ is the significant level for testing $H_0$ in \eqref{H0}.
The PTZIB estimator has two options, when $H_0$ is accepted, $\hat{\boldsymbol{\theta}}_{PTZIB}= \hat{\boldsymbol{\theta}}_{REZIB}$ and when $H_0$ is rejected $\hat{\boldsymbol{\theta}}_{PTZIB}= \hat{\boldsymbol{\theta}}_{UNZIB}$.

\section{Properties of Proposed Estimators}\label{sec4}
In the previous section, we studied the asymptotic distributional properties of the proposed estimators when the null hypothesis $H_0$ is incorrect. Suppose that $\boldsymbol{R}\boldsymbol{\theta} = \boldsymbol{r} + \boldsymbol{\zeta}$. For any fixed $\boldsymbol{\zeta}$, the test statistic $T_n$ converges to infinity, and thus, all the proposed estimators will asymptotically converge to the UNZIB estimator. Consequently, we examine the asymptotic properties of the proposed estimators under the subsequent sequence of local alternatives:
\begin{equation}\label{Hn}
H'_{0}: ~ \boldsymbol{R}\boldsymbol{\theta}=\boldsymbol{r}+\dfrac{\boldsymbol{\vartheta}}{\sqrt{n}}
\end{equation}
where $ \boldsymbol{\vartheta}=(\vartheta_1, \ldots,\vartheta_{p_2})^T \in \mathbb{R}^{p_2} $ is a fixed vector with real values. The definition of $H'_{0}$ shows that
null-hypothesis in \eqref{H0} is a special case of \eqref{Hn} and the value of $ \boldsymbol{\vartheta}/\sqrt{n} $ will be the distance between the value of   $\boldsymbol{R}\boldsymbol{\theta}$ in \eqref{Hn} and \eqref{H0}.
\\
Let the asymptotic cumulative distribution of $\boldsymbol{\theta}$ under null-hypothesis in \eqref{Hn} as
\begin{equation*}
G_n (\boldsymbol{\theta})= \lim\limits_{n \to \infty} p\big(\sqrt{n}(\tilde{\boldsymbol{\theta}}-\boldsymbol{\theta}) \vert H'_{0}\big)
\end{equation*}
Hence, the asymptotic distributional bias of $\tilde{\boldsymbol{\theta}}$ is definded as:
\begin{equation}\label{bias}
ADB(\tilde{\boldsymbol{\theta}})= \lim\limits_{n \to \infty} \mathbb{E} \big[ \sqrt{n}(\tilde{\boldsymbol{\theta}}-\boldsymbol{\theta})\vert H'_{0} \big]=\int \ldots \int t \, dG_n(t)
\end{equation}
and the asymptotic distributional mean squared error (ADMSE) of estimator $\tilde{\boldsymbol{\theta}}$ is defined as:
\begin{equation}\label{cov}
ADMSE(\tilde{\boldsymbol{\theta}})= \lim\limits_{n \to \infty} \mathbb{E} \big[ \sqrt{n}(\tilde{\boldsymbol{\theta}}-\boldsymbol{\theta}) \sqrt{n}(\tilde{\boldsymbol{\theta}}-\boldsymbol{\theta})^T\vert H'_{0}\big] = \int \ldots \int tt^T \, dG_n(t)
\end{equation}
To derive the properties of the proposed estimators, we present the following lemmas.
\begin{Lemma}\label{lem1}
Consider
\begin{align*}
U_1&= \sqrt{n}(\hat{\boldsymbol{\theta}}_{UNZIB}-\boldsymbol{\theta}) \\
U_2&= \sqrt{n}(\hat{\boldsymbol{\theta}}_{REZIB}-\boldsymbol{\theta}) \\
U_3&= \sqrt{n}(\hat{\boldsymbol{\theta}}_{UNZIB}-\hat{\boldsymbol{\theta}}_{REZIB})
\end{align*}
Hence, under the null hypothesis in \eqref{Hn}, regularity conditions and when $n$ increases, we have:
\begin{equation}
 \left[ \begin{gathered}
  {U_1} \hfill \\
  {U_2} \hfill \\
  {U_3} \hfill \\
\end{gathered}  \right] \sim N_{3k} \left( {\left[ \begin{array}{c}
  \boldsymbol{0} \hfill \\
 -\boldsymbol{J}\boldsymbol{\vartheta}\hfill \\
 \boldsymbol{J}\boldsymbol{\vartheta} \hfill \\
\end{array}  \right] ,\left[ {\begin{array}{*{20}{c}}
 \boldsymbol{F}^{-1} & \boldsymbol{F}^{-1}- \boldsymbol{J} \boldsymbol{F}^{-1} &  \boldsymbol{J} \boldsymbol{F}^{-1} \\
 &  \boldsymbol{F}^{-1}- \boldsymbol{J} \boldsymbol{F}^{-1} &  \boldsymbol{0} \\
 & &  \boldsymbol{J} \boldsymbol{F}^{-1}
\end{array}}\right]}\right)
\end{equation}
where $\boldsymbol{J}= \boldsymbol{F}^{-1} \boldsymbol{R}^T (\boldsymbol{R} \boldsymbol{F}^{-1}\boldsymbol{R}^T)^{-1}\boldsymbol{R} $.
\end{Lemma}
\begin{proof}
See Appendix \ref{plem1}.
\end{proof}

\begin{Lemma}\label{lem2}
Under the null hypothesis in \eqref{Hn} and regularity conditions, when $n$ increases the test statistic $T_n$ converges to chi-squared distribution with $p_2$ degree of freedoms and non-central parameter $\lambda=\boldsymbol{\vartheta}^T (\boldsymbol{R} \boldsymbol{F}^{-1}\boldsymbol{R}^T)^{-1}\boldsymbol{\vartheta}$.
\end{Lemma}
\begin{proof}
See \cite{Davidson}.
\end{proof}

\begin{Lemma}\label{lem3}
Let $\boldsymbol{U}$ be a $r$-dimensional vector follows $N_{r}(\mu _{U}, \pmb{\Sigma}_{U})$. For any measurable function of $h(.)$, we have
\begin{align*}
\mathbb{E}\bigg(\boldsymbol{U} \, h\big(\boldsymbol{U}^t \boldsymbol{U}\big)\bigg) &= \mu_{U} \mathbb{E}\bigg( h\big(\chi_{r+2}^{2}(\lambda)\big)\bigg) \\
\mathbb{E}\bigg(\boldsymbol{U}^t \boldsymbol{U}\, h \big(\boldsymbol{U}^t \boldsymbol{U}\big)\bigg) &=\pmb{\Sigma}_{U} \mathbb{E}\bigg( h\left(\chi_{r+2}^{2}(\lambda)\right)\bigg) + \mu_{Z}' \mu_{U}\mathbb{E}\bigg(h\left(\chi_{r+4}^{2}(\lambda)\right)\bigg),
\end{align*}
where $\lambda= \mu_{U}'\pmb{\Sigma}_{U}^{-1} \mu_{U}$ is the non-centrality parameter.
\end{Lemma}
\begin{proof}
See \cite{judge}.
\end{proof}
In the following theorems, we will establish the asymptotic properties of the suggested estimators.
\begin{Theorem}\label{thm1}
Under the null hypothesis stated in equation \eqref{Hn} and assuming regularity conditions hold,  the asymptotic distributional bias of the suggested estimators is as follows:
\begin{align*}
ADB(\hat{\boldsymbol{\theta}}^{REZIB}) &= -\boldsymbol{J \vartheta}\\
ADB(\hat{\boldsymbol{\theta}}^{JSZIB}) &= -(p_2-2) \boldsymbol{J \vartheta} \mathbb{E}(\chi^{-2}_{p_2} (\lambda))\\
ADB(\hat{\boldsymbol{\theta}}^{PJSZIB}) &=  -(p_2-2) \boldsymbol{J \vartheta}\bigg\{\mathbb{E}\bigg[ \chi^{-2}_{q+2} (\lambda) \bigg]-\mathbb{E}\bigg[\chi^{-2}_{p_2+2} (\lambda)I_{(\chi^{-2}_{p_2+2} (\lambda)< p_2-2)} \bigg]\bigg\}\\
&\qquad \qquad - \boldsymbol{J \vartheta}  \Phi_{p_2+2} (p_2-2; \lambda)\\
ADB(\hat{\boldsymbol{\theta}}^{PTZIB}) &= -\boldsymbol{J \vartheta}\Phi_{p_2} \big(\chi^2_{(p_2, \alpha)}; \lambda\big)
\end{align*}
where  $ \Phi_c \big(a; \lambda\big) = p(\chi^{2}_c (\lambda)\leq  a)$ and $\mathbb{E}(\chi^{2k}_c (\lambda))$ is the $k$th order moment of a non-central $\chi^2$ distribution with $c$ degrees of freedom and non-central parameter $\lambda= \boldsymbol{\vartheta}^T(\boldsymbol{R} \boldsymbol{F}^{-1}\boldsymbol{R}^T)^{-1}\boldsymbol{\vartheta}$.
\end{Theorem}
\begin{proof}
See Appendix \ref{proof1}.
\end{proof}
We present the following theorem regarding the ADMSE of the estimators:
\begin{Theorem}\label{thm2}
Under the null hypothesis stated in equation \eqref{Hn} and assuming regularity conditions hold, the ADMSE of the suggested estimators is as follows:
\begin{align*}
ADMSE(\hat{\boldsymbol{\theta}}_{UNZIB}) &= \boldsymbol{F}^{-1}\\
ADMSE(\hat{\boldsymbol{\theta}}_{REZIB}) &=  \boldsymbol{F}^{-1}- \boldsymbol{J}\boldsymbol{F}^{-1} +  \boldsymbol{W}\\
ADMSE(\hat{\boldsymbol{\theta}}_{JSZIB}) &=\boldsymbol{F}^{-1}+ 2 (p_2-2) \boldsymbol{W} \bigg\{ \mathbb{E}\bigg[ \chi^{-2}_{q+2} (\lambda)\bigg]-\mathbb{E}\bigg[ \chi^{-2}_{q+4} (\lambda) \bigg] +\dfrac{p_2-2}{2}\mathbb{E}\bigg[ \chi^{-4}_{q+4} (\lambda) \bigg] \bigg\}\nonumber\\
& \qquad -2 (p_2-2)\boldsymbol{J}\boldsymbol{F}^{-1} \bigg\{ \mathbb{E}\bigg[ \chi^{-2}_{q+4} (\lambda) \bigg] -\dfrac{p_2-2}{2}\mathbb{E}\bigg[\chi^{-4}_{q+2} (\lambda)\bigg]\bigg\}\\
ADMSE(\hat{\boldsymbol{\theta}}_{PJSZIB}) &=ADMSE(\hat{\boldsymbol{\theta}}^{JSZIB}) - 2 \boldsymbol{W} \Phi_{p_2+2}(T_{(p_2, \alpha)}; \lambda)\nonumber \\
& \qquad-\boldsymbol{J}\boldsymbol{F}^{-1}\mathbb{E}\bigg[\bigg(1- (p_2-2)\chi^{-2}_{p_2+2} (\lambda)\bigg)^2 I_{\big(\chi^{2}_{p_2+2} (\lambda)< p_2-2\big)}\bigg]\nonumber \\
& \qquad - \boldsymbol{W} \mathbb{E}\bigg[\bigg(1- (p_2-2)\chi^{-2}_{p_2+4} (\lambda)\bigg)^2 I_{\big(\chi^{2}_{p_2+4} (\lambda)< p_2-2\big)}\bigg]\nonumber \\
ADMSE(\hat{\boldsymbol{\theta}}_{PTZIB}) &=\boldsymbol{F}^{-1} + \boldsymbol{J}\boldsymbol{F}^{-1}\Phi_{p_2+2}(T_{(p_2, \alpha)}; \lambda)+\boldsymbol{W}\Phi_{p_2+4}(T_{(p_4, \alpha)}; \lambda) \nonumber
\end{align*}
where $\boldsymbol{W}=\boldsymbol{J\vartheta}\boldsymbol{\vartheta}^T \boldsymbol{J}^T$
\end{Theorem}
\begin{proof}
See Appendix \ref{proof2}.
\end{proof}

\section{Simulation Study}\label{sec5}
In this section, we will perform a simulation study to evaluate the performance of the suggested estimators discussed in the previous section. The benchmark for comparison will be the simulated relative efficiency, which is defined as follows:
\begin{equation}\label{SRE}
SRE(\hat{\boldsymbol{\theta}})= \dfrac{SMSE(\hat{\boldsymbol{\theta}}_{UNZIB})} {SMSE(\hat{\boldsymbol{\theta}})}
\end{equation}
where $\hat{\boldsymbol{\theta}}$ is one of the suggested estimators in this paper.
The data are generated from a ZIBell regression such that:
\begin{equation*}
\mu_i= \exp \{\beta_1 x_{i1}+ \beta_2 x_{i2}+ \ldots + \beta_p x_{ip} \} \qquad i=1, 2, \ldots, n \quad \text{and} \quad p=5, 7
\end{equation*}
and
\begin{equation*}
logit(\pi_i)=\log(\dfrac{\pi_i}{1-\pi_i})= \gamma_1 z_{i1}+ \gamma_2 z_{i2}+ \gamma_3 z_{i3}   \qquad i=1, 2, \ldots, n
\end{equation*}
where $x_{i1}=z_{i1}=1$ to consider a model with intercepts. The covariates observations for $\mu_i$ are generated from standard normal distribution and for $\pi_i$ are generated from standard exponential distribution. The real value of the parameters are considered as
\begin{align}
\boldsymbol{\beta}^T&= (0.5, 1, -1.5, \underbrace{0, \ldots, 0}_{p-3})^T \label{breal}\\
 \boldsymbol{\gamma}^T&= ( 0.5, -1, 0)^T\label{greal}
\end{align}
The these values  indicates that $\beta_4, \ldots, \beta_p$ and $\gamma_3$ are the inactive parameters in our model. Therefore, the hypothesis matrix is made by using these parameters, such that
\begin{equation*}
\boldsymbol{H}_{(p-2)\times (p+3)}=\begin{pmatrix}
 0& 0& 0& 1& 0& \ldots &0& 0& 0 & 0& 0\\
 0& 0& 0& 0& 1& \ldots &0& 0& 0 & 0& 0\\
 \vdots & \vdots & \vdots &\vdots & \vdots & \ddots &\vdots &\vdots &\vdots & \vdots & \vdots\\
 0& 0& 0& 0& 0& \ldots &1& 0& 0 & 0& 0\\
 0& 0& 0& 0& 0& \ldots &0& 0& 0 & 0& 1\\
\end{pmatrix}
, \quad \boldsymbol{h}_{(p-2)\times 1}
=
\begin{pmatrix}
0\\
0\\
\vdots\\
0\\
0
\end{pmatrix}
\end{equation*}
For the next step, we define $\delta^2= \| \boldsymbol{\theta}- \boldsymbol{\theta}^{real} \|$ where $\| . \|$ is the Euclidean norm. Due to this definition,  $\delta^2$ is the departure measure of $\boldsymbol{\theta}$ from its real value. To understand the effect of departure measure, we consider
\begin{equation*}
\boldsymbol{\beta}^T= (0.5, 1, -1.5, \underbrace{0, \ldots, 0}_{p-3})^T \quad \text{and} \quad \boldsymbol{\gamma}^T= ( 0.5, -1, \delta)^T
\end{equation*}
and the real value of parameters is defined in \eqref{breal} and \eqref{greal}. Thus, when $\delta^2=0$ means the real value of $\boldsymbol{\theta}=0$ but as $\delta^2$ increases, we get away from the real value of $\boldsymbol{\theta}$.
We assumed different values of $\delta^2$ to be a sequence of numbers from zero to two with a step length of 0.2.
On the other hand, for different scenarios, we choose $n=50, 100, 200$ and $p=5, 7$.

We repeated each scenario $1000$ times to calculate the value of SMSE in \eqref{SRE} by the following definition:
\begin{equation*}
SMSE(\hat{\boldsymbol{\theta}})= \dfrac{1}{1000} \sum_{r=1}^{1000} (\hat{\boldsymbol{\theta}}_r-\boldsymbol{\theta})^T(\hat{\boldsymbol{\theta}}_r-\boldsymbol{\theta})
\end{equation*}
where $ \hat{\boldsymbol{\theta}}_r $ is the estimated value of $ \boldsymbol{\theta} $ at the $r$th repetition. The definition of SRE shows that the $\hat{\boldsymbol{\theta}}$ outperformed the $\hat{\boldsymbol{\theta}}^{UN}$ if SRE is larger than one.
The definition of SRE indicates that $\hat{\boldsymbol{\theta}}$ performs better than $\hat{\boldsymbol{\beta}}_{UNZIB}$ whenever the SRE value is larger than one. The results of the simulation are presented in Table \ref{Tab1}. For ease of comparison, the results are presented graphically in Figure \ref{fig1}.
Based on the table presented, the following conclusions can be drawn:
\begin{itemize}
\item When the value of $\delta^2$ is zero, the REZIB's performance is optimal in all scenarios. However, as the value of $\delta^2$ increases, the SRE of the REZIB decreases sharply, eventually becoming inefficient.
\item When $\delta^2$ is zero, the PJSZIB outperforms the JSZIB. Although, as the value of $\delta^2$ increases, the performance of both estimators approaches one, but the PJSZIB remains superior to the JSZIB.
\item The PTZIB outperforms the James-Stein type estimators when the $\delta^2=0$.
\item Increasing the number of covariates generally leads to an increase in the SRE of all estimators, regardless of the value of $\delta^2$.
\item Generally, as the sample size increases, the SRE of all estimators tends to increase.
\end{itemize}

\begin{figure}
\centering
\includegraphics[scale=0.85]{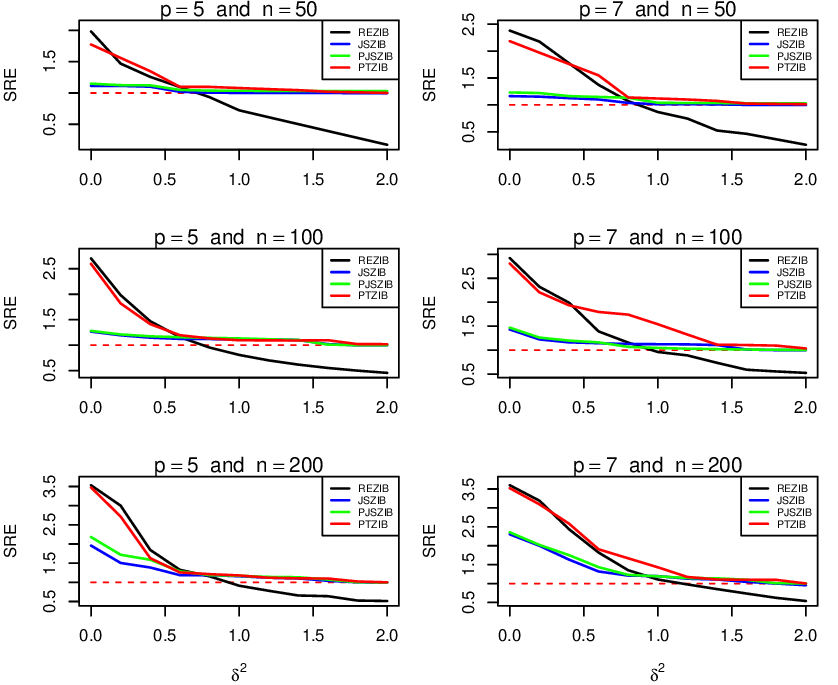} 
\caption{The SRE of suggested estimators with different values of $\delta^2$.}\label{fig1}
\end{figure}

\begin{table}[h]
    \caption{The SRE of proposed estimators.}\label{Tab1}
\begin{tabular}{ccccccccccc}
\toprule
        &   & \multicolumn{4}{c}{$p=5$}  && \multicolumn{4}{c}{$p=7$}\\
        \cmidrule{3-6} \cmidrule{8-11}
$n$ & $\delta^2$ &\text{REZIB} & \text{JSZIB}& \text{PJSZIB} & \text{PTZIB} &&  \text{REZIB} & \text{JSZIB}& \text{PJSZIB} & \text{PTZIB}\\
   \hline
50 & 0.0 &1.9828 &1.1137 &1.1491 &1.9741 && 2.3795 &1.1636 &1.2316 &2.1828\\
    & 0.2 &1.4692 &1.1136 &1.1248 &1.7600 && 2.1757 &1.1531 &1.2205 &1.9703\\
    & 0.4 &1.2571 &1.0998 &1.1217 &1.3483 && 1.9734 &1.1265 &1.1636 &1.7591\\
    & 0.6 &1.0955 &1.0188 &1.0492 &1.0976 && 1.3715 &1.1014 &1.1479 &1.5483\\
    & 0.8 &0.9342 &1.0041 &1.0371 &1.0972 && 1.0697 &1.0371 &1.1331 &1.1381\\
    & 1.0 &0.7233 &1.0015 &1.0353 &1.0788 && 0.8680 &1.0027 &1.0389 &1.1191\\
    & 1.2 &0.6126 &1.0007 &1.0337 &1.0615 && 0.7663 &1.0018 &1.0352 &1.1005\\
    & 1.4 &0.5022 &0.9999 &1.0324 &1.0452 && 0.5247 &1.0003 &1.0349 &1.0742\\
    & 1.6 &0.3921 &0.9990 &1.0312 &1.0198 && 0.4632 &0.9999 &1.0319 &1.0252\\
    & 1.8 &0.2822 &0.9980 &1.0301 &1.0111 && 0.3617 &0.9997 &1.0304 &1.0184\\
    & 2.0 &0.1725 &0.9970 &1.0290 &1.0007 && 0.2602 &0.9983 &1.0293 &1.0134  \\
[5pt]
100 & 0.0 &2.7034 &1.2658 &1.2806 &2.5968 && 2.9205 &1.4302 &1.4686 &2.8064\\
    & 0.2 &1.9821 &1.1936 &1.2105 &1.8196 && 2.3225 &1.2231 &1.2617 &2.2060\\
    & 0.4 &1.4700 &1.1469 &1.1741 &1.4135 && 1.9905 &1.1613 &1.1980 &1.9352\\
    & 0.6 &1.1590 &1.1225 &1.1559 &1.1957 && 1.3899 &1.1438 &1.1614 &1.7965\\
    & 0.8 &0.9534 &1.1198 &1.1428 &1.1310 && 1.1562 &1.1265 &1.0674 &1.7409\\
    & 1.0 &0.8083 &1.1169 &1.1331 &1.0963 && 0.9609 &1.1234 &1.0453 &1.5400\\
    & 1.2 &0.7008 &1.1125 &1.1166 &1.0937 && 0.8896 &1.1196 &1.0297 &1.3241\\
    & 1.4 &0.6180 &1.1010 &1.1029 &1.0924 && 0.7343 &1.1061 &1.0196 &1.1136\\
    & 1.6 &0.5524 &1.0198 &1.0104 &1.0974 && 0.5901 &1.0102 &1.0119 &1.1056\\
    & 1.8 &0.4993 &0.9987 &0.9994 &1.0204 && 0.5541 &0.9995 &1.0060 &1.0938\\
    & 2.0 &0.4553 &0.9985 &1.0006 &1.0191 && 0.5241 &0.9986 &1.0020 &1.0342\\
    [5pt]
200 & 0.0 &3.5308 &1.9577 &2.1802 &3.4737 && 3.5993 &2.3020 &2.3553 &3.5208\\
    & 0.2 &2.9998 &1.5052 &1.7188 &2.7154 && 3.1916 &1.9996 &2.0213 &3.0951\\
    & 0.4 &1.4455 &1.3861 &1.5885 &1.6304 && 2.4374 &1.6329 &1.7513 &2.5840\\
    & 0.6 &1.3260 &1.1873 &1.2867 &1.2592 && 1.8195 &1.3195 &1.4332 &1.9057\\
    & 0.8 &1.2568 &1.1875 &1.1883 &1.2137 && 1.3515 &1.2101 &1.2292 &1.6673\\
    & 1.0 &0.9118 &1.1588 &1.1699 &1.1796 && 1.1074 &1.1972 &1.1913 &1.4301\\
    & 1.2 &0.7801 &1.1211 &1.1413 &1.1194 && 0.9764 &1.1277 &1.1454 &1.1732\\
    & 1.4 &0.6567 &1.1021 &1.1331 &1.1041 && 0.8534 &1.1067 &1.1398 &1.1080\\
    & 1.6 &0.6386 &1.0398 &1.0721 &1.0987 && 0.7358 &1.0441 &1.0994 &1.0990\\
    & 1.8 &0.5243 &1.0002 &1.0002 &1.0217 && 0.6217 &1.0002 &1.0072 &1.0981\\
    & 2.0 &0.5126 &0.9993 &0.9998 &1.0002 && 0.5383 &0.9554 &1.0054 &1.0012\\
\bottomrule
\end{tabular}
\end{table}

\section{Application to a Real Dataset}\label{sec6}
In this section, we apply the proposed estimators in ZIBell models to analyze a real dataset, specifically the wildlife fish dataset \cite{Saffari}. This dataset comprises 250 observations with 5 covariates: $X_1$ represents whether the trip was not just for fishing (nofish), 0 if no 1 and if yes,
$X_2$ represents whether live bait was used or not (livebait), 0 if no and 1 if yes, $X_3$ represents whether or not they brought a camper (camper), $X_4$ represents how many total persons on the trip (persons) and $X_5$ represents how many children present (child).  The response variable $(y)$ represents the number of fish caught. Algamal et al. \cite{Algamal2} have demonstrated that the ZIBell regression model is suitable for this dataset. Following this conclusion, we consider the following full model for this dataset:
\begin{align}
\log (\mu_i)&= \beta_0+ \beta_1 X_{i1} + \beta_2 X_{i2}+\beta_3 X_{i3}+ \beta_4 X_{i4}+ \beta_5 X_{i5} \qquad i=1, 2, \dots, 250\\
\log (\dfrac{\pi_i}{1- \pi_i})&= \gamma_0+ \beta_1 X_{i1} + \gamma_2 X_{i2}+ \gamma_3 X_{i3}+ \gamma_4 X_{i4}+ \gamma_5 X_{i5}
\end{align}
We fit this model and results show that for $\mu_i$, only $X_1$ and $X_2$ and for $\pi_i$ all covariates are significant. Thus, we use this information to conduct our sub-model. The restrictions on the model will be:
\begin{align*}
\beta_j=0,  \quad j= 3, 4, 5
\end{align*}
The estimates by using different estimators are reported in Table \ref{tab2}. To compare the estimators, we apply the following formulas to find the MSE and the mean absolute error for any estimators:
\begin{align*}
MSE(\hat{\boldsymbol{\theta}})&= \dfrac{1}{250} \sum_{i=1}^{250} \big(y_i- (1-\hat{\pi}_i)\hat{\mu}_i\big)\\
MAE(\hat{\boldsymbol{\theta}})&= \dfrac{1}{250} \sum_{i=1}^{250} \big\vert y_i- (1-\hat{\pi}_i)\hat{\mu}_i\big\vert
\end{align*}
where $ \hat{\boldsymbol{\theta}} $ can be any proposed estimators and $ \hat{\pi}_i $ and $\hat{\mu}_i$ are obtained  using \eqref{mui} and \eqref{pii}, respectively.
Based on the table, the proposed shrinkage estimators outperform the UNZIB estimator in terms of MSE and MAE. However, the REZIB and PTZIB perform better than the Stein-type estimators.

\begin{table}[h]
    \caption{Estimates, MSE and MAE by proposed estimators for wildlife fish dataset.}\label{tab2}
\begin{tabular}{cccccc}
\toprule
\text{Parameters} & \text{UNZIB} &\text{REZIB} & \text{JSZIB}& \text{PJSZIB} & \text{PTZIB} \\
\hline
$\beta_0 $ &-1.7939 &-1.7856& -1.7936& -1.7936& -1.7939\\
$\beta_1$  &-0.7437& -0.7355& -0.7434& -0.7434& -0.7437\\
$\beta_2$  &1.1757&  1.1672&  1.1754&  1.1754&  1.1757\\
$\beta_3$  &0.4642&  0.5278&  0.4666&  0.4666&  0.4642\\
$\beta_4$  &0.8797 & 0.8727&  0.8795&  0.8795&  0.8797\\
$\beta_5$  &-1.1753& -1.3278& -1.1809& -1.1809& -1.1753\\[5pt]
$\gamma_0$  &1.5709 &-0.4733  &1.4960  &1.4960  &1.5709\\
$\gamma_1$  &-1.6335 &-0.7970 &-1.6028& -1.6028& -1.6335\\
$ \gamma_2$ &-0.3437  &0.2394& -0.3224& -0.3224& -0.3437\\
$\gamma_3 $ &-1.0164 & 0.0000& -0.9792& -0.9792& -1.0164\\
$\gamma_4$  &-0.8286  &0.0000& -0.7982& -0.7982& -0.8286\\
$\gamma_5$  &2.0524  &0.0000&  1.9771&  1.9771&  2.0524\\
\bottomrule
\text{MSE}& 129.6031 & 126.1733 & 128.4942 &128.4942 &126.1733\\
\text{MAE}& 4.8105 & 2.7865 & 3.8096 & 3.8096 & 2.7865 \\
\bottomrule
\end{tabular}
\end{table}

\section{Conclusion}\label{sec7}
In this paper, we propose the Stein-type shrinkage and preliminary test estimators by combining the UNZIB and REZIB estimators in ZIBell regression models based on some information on models parameters. The asymptotic distributional bias and mean squared error of the suggested estimators are obtained.
A simulation study is conducted, and real-world data are analyzed to show the performance of our proposed shrinkage estimators.
Our findings reveal that the proposed shrinkage estimator leads to a considerable reduction in mean squared error when the information on parameters is well-defined.

\begin{appendices}

\section{Proof of Lemma \ref{lem1}}\label{plem1}
according to \eqref{distun}, the mean and covariance matrix of $U_1$ is clear. Now, due to the definition of $U_2$ and \eqref{RES}, we have:
\begin{align}\label{rezib}
U_2 & = \sqrt{n}(\hat{\boldsymbol{\theta}}_{REZIB}-\boldsymbol{\theta}) \nonumber\\
&= \sqrt{n}\bigg(\hat{\boldsymbol{\theta}}_{UNZIB} - \boldsymbol{F}^{-1}  \boldsymbol{R}^T \big[\boldsymbol{R} \boldsymbol{F}^{-1}\boldsymbol{R}^T\big]^{-1} (\boldsymbol{R}\hat{\boldsymbol{\theta}}_{UNZIB}- \boldsymbol{r})-\boldsymbol{\theta}\bigg) \nonumber\\
&= \sqrt{n} \big(\hat{\boldsymbol{\theta}}_{UNZIB}-\boldsymbol{\theta}\big) -\sqrt{n} \boldsymbol{J} \big(\hat{\boldsymbol{\theta}}_{UNZIB}-\boldsymbol{\theta}\big) -\boldsymbol{J} \sqrt{n} \big(\boldsymbol{R}\boldsymbol{\theta} -\boldsymbol{r}\big) \nonumber\\
& = U_1 - \boldsymbol{J}U_1-\boldsymbol{J}\boldsymbol{\vartheta}
\end{align}
Thus, the mean and covariance matrix of $U_2$ will be:
\begin{align}
\mathbb{E}(U_2)&= \mathbb{E}\bigg(U_1 - \boldsymbol{J}U_1-\boldsymbol{J}\boldsymbol{\vartheta}\bigg)= -\boldsymbol{J}\boldsymbol{\vartheta} \\
cov(U_2)&= cov(U_1 - \boldsymbol{J}U_1+\boldsymbol{J}\boldsymbol{\vartheta})\nonumber\\
&= cov(U_1)+\boldsymbol{J}cov(U_1)\boldsymbol{J}^T-2cov(U_1, \boldsymbol{J}U_1)\nonumber\\
&= \boldsymbol{F}^{-1}+\boldsymbol{J} \boldsymbol{F}^{-1} \boldsymbol{J}^T-2 \boldsymbol{F}^{-1}\boldsymbol{J}^T\nonumber\\
&=  \boldsymbol{F}^{-1}-\boldsymbol{J} \boldsymbol{F}^{-1}
\end{align}
We can rewrite $U_3$ by using $U1$ and $U_2$ as:
\begin{align}
U_3 & =  \sqrt{n}(\hat{\boldsymbol{\theta}}_{UNZIB}-\hat{\boldsymbol{\theta}}_{REZIB})  \nonumber\\
&= \sqrt{n}\bigg(\big(\hat{\boldsymbol{\theta}}_{UNZIB}- \boldsymbol{\theta} \bigg)-\big(\hat{\boldsymbol{\theta}}_{REZIB}-\boldsymbol{\theta}\big) \bigg)  \nonumber\\
& = U_1-(U_1 - \boldsymbol{J}U_1-\boldsymbol{J}\boldsymbol{\vartheta})\nonumber\\
& = \boldsymbol{J}U_1+\boldsymbol{J}\boldsymbol{\vartheta}
\end{align}
Hence, the mean and covariance matrix of $U_3$ will be:
\begin{align}
\mathbb{E}(U_3)&= \mathbb{E}\bigg(\boldsymbol{J}U_1+\boldsymbol{J}\boldsymbol{\vartheta}\bigg)= \boldsymbol{J}\boldsymbol{\vartheta} \\
cov(U_3)&= cov(\boldsymbol{J}U_1+\boldsymbol{J}\boldsymbol{\vartheta})= \boldsymbol{J}cov(U_1)\boldsymbol{J}^T= \boldsymbol{J} \boldsymbol{F}^{-1}
\end{align}
For covariance of pairwise of variables $U_1$, $U_2$ and $U_3$, we have:
\begin{align}
cov(U_1, U_2) & =  cov(U_1, \boldsymbol{J}U_1+\boldsymbol{J}\boldsymbol{\vartheta})  \nonumber\\
&= cov(U_1)- \boldsymbol{J} cov(U_1) =  \boldsymbol{F}^{-1}-\boldsymbol{J} \boldsymbol{F}^{-1}
\end{align}
and
\begin{align}
cov(U_1, U_3) & =  cov(U_1, \boldsymbol{J}U_1+\boldsymbol{J}\boldsymbol{\vartheta})  \nonumber\\
&= \boldsymbol{J} cov(U_1) =\boldsymbol{J} \boldsymbol{F}^{-1}
\end{align}
and
\begin{align}
cov(U_2, U_3) & =  cov(\boldsymbol{J}U_1-\boldsymbol{J}\boldsymbol{\vartheta}, \boldsymbol{J}U_1+\boldsymbol{J}\boldsymbol{\vartheta})  \nonumber\\
&= \boldsymbol{J} cov(U_1) \boldsymbol{J}^T = \boldsymbol{J} \boldsymbol{F}^{-1} \boldsymbol{J}^T =  \boldsymbol{J} \boldsymbol{F}^{-1}
\end{align}

\vspace{10pt}
\section{Proof of Theorem \ref{thm1}}\label{proof1}
By using Lemma \ref{lem1}, we have:
\begin{equation}
ADB(\hat{\pmb{\beta}}^{REZIB})= \mathbb{E}\big(\sqrt{n}(\hat{\boldsymbol{\theta}}_{REZIB}-\boldsymbol{\theta})\big)= \mathbb{E}(U_2)= -  \boldsymbol{J \vartheta}.
\end{equation}
and
\begin{align}
ADB\big(\hat{\boldsymbol{\theta}}_{JSZIB}\big)& =\mathbb{E}\big(\sqrt{n}(\hat{\boldsymbol{\theta}}_{JSZIB}-\boldsymbol{\theta})\big) \nonumber \\
&= \mathbb{E} \bigg[\sqrt{n} \bigg\{(\hat{\boldsymbol{\theta}}_{REZIB}- \boldsymbol{\theta})+ \big(1-\dfrac{p_2-2}{T_n}\big)\big(\hat{\boldsymbol{\theta}}_{UNZIB}-\hat{\boldsymbol{\theta}}_{REZIB} \big) \bigg\} \bigg] \nonumber\\
& =\mathbb{E} \bigg[ U_2 + \big(1-\dfrac{p_2-2}{T_n}\big) U_3 \bigg] \nonumber \\
& = \mathbb{E} \big(U_2  \big) + \mathbb{E}\bigg[ \big(1-\dfrac{p_2-2}{T_n}\big) U_3 \bigg] \nonumber
\end{align}
By using Lemma \eqref{lem1} to Lemma \eqref{lem3}, we have:
\begin{align}
\mathbb{E}\big(\sqrt{n}(\hat{\boldsymbol{\theta}}_{JSZIB}-\boldsymbol{\theta})\big)  &= -\boldsymbol{J \vartheta} + \boldsymbol{J \vartheta} \mathbb{E}\bigg[ 1-\dfrac{p_2-2}{\chi^{2}_{p_2+2} (\lambda)} \bigg]  \nonumber \\
& = -(p_2-2) \boldsymbol{J \vartheta}  \mathbb{E}\bigg[ \chi^{-2}_{p_2+2} (\lambda) \bigg]
\end{align}
For $ADB \big(\hat{\boldsymbol{\theta}}_{PJSZIB}\big)$, we derive it as:
\begin{align*}
ADB \big(\hat{\boldsymbol{\theta}}_{PJSZIB}\big)& =\mathbb{E}\big(\sqrt{n}(\hat{\boldsymbol{\theta}}_{PJSZIB}-\boldsymbol{\theta})\big)\\
&= \mathbb{E} \bigg[\sqrt{n} \bigg\{\hat{\boldsymbol{\theta}}_{REZIB}+ \bigg(1- \dfrac{p_2-2}{T_n}\bigg)^{+}
(\hat{\boldsymbol{\theta}}_{REZIB}-\hat{\boldsymbol{\theta}}_{UNZIB})\big)- \boldsymbol{\theta} \bigg\} \bigg]
\end{align*}

\begin{align}
ADB \big(\hat{\boldsymbol{\theta}}_{PJSZIB}\big)& = \mathbb{E} \bigg[\sqrt{n} \bigg\{(\hat{\boldsymbol{\theta}}_{JSZIB}- \boldsymbol{\theta})-\bigg(1- \dfrac{p_2-2}{T_n}\bigg) (\hat{\boldsymbol{\theta}}_{REZIB}-\hat{\boldsymbol{\theta}}_{UNZIB})\big) I_{(T_n < p_2-2)} \bigg]\nonumber\\
& =  \mathbb{E} \bigg[\sqrt{n} (\hat{\boldsymbol{\theta}}_{JSZIB}- \boldsymbol{\theta})- \sqrt{n}\big(\hat{\boldsymbol{\theta}}_{UNZIB}-\hat{\boldsymbol{\theta}}_{REZIB} \big)\bigg(1- \dfrac{p_2-2}{T_n}\bigg) I_{(T_n< p_2-2)} \bigg] \nonumber\\
& =\mathbb{E} \bigg[\sqrt{n} (\hat{\boldsymbol{\theta}}_{JSZIB}- \boldsymbol{\theta})- U_3 I_{(T_n< p_2-2)}+(p_2-2) U_3 T_n^{-1}I_{(T_n< p_2-2)}\bigg] \nonumber\\
& = ADB \big(\hat{\boldsymbol{\theta}}_{JSZIB}\big) - \mathbb{E}\bigg[U_3 I_{(T_n< p_2-2)} \bigg] + (p_2-2) \mathbb{E}\bigg[U_3 T_n^{-1}I_{(T_n< p_2-2)} \bigg] \nonumber \\
& = ADB \big(\hat{\boldsymbol{\theta}}_{JSZIB}\big) - \mathbb{E} (U_3) \mathbb{E}\bigg[I_{(T_n< p_2-2)}\bigg] +(p_2-2) \mathbb{E} (U_3) \mathbb{E}\bigg[T_n^{-1}I_{(T_n< p_2-2)} \bigg]\nonumber \\
& =  ADB \big(\hat{\boldsymbol{\theta}}_{JSZIB}\big) - \boldsymbol{J \vartheta}  \Phi_{p_2+2} (p_2-2; \lambda) +(p_2-2)\boldsymbol{J \vartheta} \mathbb{E}\bigg[\chi^{-2}_{p_2+2} (\lambda)I_{(\chi^{-2}_{p_2+2} (\lambda)< p_2-2)} \bigg]\nonumber \\
& = -(p_2-2) \boldsymbol{J \vartheta}\bigg\{\mathbb{E}\bigg[ \chi^{-2}_{p_2+2} (\lambda) \bigg]-\mathbb{E}\bigg[\chi^{-2}_{p_2+2} (\lambda)I_{(\chi^{-2}_{p_2+2} (\lambda)< p_2-2)} \bigg]\bigg\}- \boldsymbol{J \vartheta}  \Phi_{p_2+2} (p_2-2; \lambda) \nonumber.
\end{align}
Finally, for the ADB of $\hat{\boldsymbol{\theta}}_{PTZIB}$, we have:
\begin{align}
ADB \big(\hat{\boldsymbol{\theta}}_{PTZIB}\big)& =\mathbb{E}\big(\sqrt{n}(\hat{\boldsymbol{\theta}}_{PTZIB}-\boldsymbol{\theta})\big) \nonumber \\
& = \mathbb{E} \bigg[\sqrt{n} \bigg\{(\hat{\boldsymbol{\theta}}_{UNZIB}- \boldsymbol{\theta}) - (\hat{\boldsymbol{\theta}}_{UNZIB}
-\hat{\boldsymbol{\theta}}_{REZIP}) I_{(T_n \leq T_{(p_2, \alpha)})}\bigg\} \bigg] \nonumber\\
& =\mathbb{E} \bigg[U_1- U_3 I_{(T_n \leq T_{(p_2, \alpha)})} \bigg] \nonumber\\
& = \mathbb{E}  (U_1)- \mathbb{E} \bigg[U_3 I_{(T_n \leq T_{(p_2, \alpha)})} \bigg] \nonumber\\
& =- \mathbb{E}  (U_3)\mathbb{E} \bigg[I_{(T_n \leq T_{(p_2, \alpha)})} \bigg] \nonumber\\
& = - \boldsymbol{J\vartheta} \Phi_{p_2+2}(T_{(p_2, \alpha)}; \lambda)\nonumber
\end{align}

\vspace{10pt}
\section{Proof of Theorem \ref{thm2}}\label{proof2}
By definition of the ADMSE, we have
\begin{align}
ADMSE(\hat{\boldsymbol{\theta}}_{UNZIB}) &=\mathbb{E}\bigg[ \sqrt{n}(\hat{\boldsymbol{\theta}}_{UNZIB}-\boldsymbol{\theta})\sqrt{n}(\hat{\boldsymbol{\theta}}_{UNZIB}-\boldsymbol{\theta})^T \bigg] =  \mathbb{E}\bigg[ U_1 U_1^T \bigg] =   \boldsymbol{F}^{-1} \nonumber
\end{align}
and
\begin{align}
ADMSE(\hat{\boldsymbol{\theta}}_{REZIB}) &=\mathbb{E}\bigg[ \sqrt{n}(\hat{\boldsymbol{\theta}}_{REZIB}-\boldsymbol{\theta})\sqrt{n}(\hat{\boldsymbol{\theta}}_{REZIB}-\boldsymbol{\theta})^T \bigg]=\mathbb{E}\bigg[ U_2 U_2^T \bigg] \nonumber\\
& =  cov(U_2)+  \bigg[\mathbb{E}(U_2)\big]\big[\mathbb{E}(U_2)\bigg]^T\nonumber \\
& =  \boldsymbol{F}^{-1}- \boldsymbol{J} \boldsymbol{F}^{-1} +  \boldsymbol{W}
\end{align}
For ADMSE of JSZIB, we have:
\begin{align}  \label{e1}
ADMSE(\hat{\boldsymbol{\theta}}_{JSZIB})&=\mathbb{E}\bigg[ \sqrt{n}(\hat{\boldsymbol{\theta}}_{JSZIB}-\boldsymbol{\theta})\sqrt{n}(\hat{\boldsymbol{\theta}}_{JSZIB}-\boldsymbol{\theta})^T \bigg] \nonumber\\
& =  \mathbb{E}\bigg[ \bigg(U_2+ \big(1-\dfrac{p_2-2}{T_n}\big) U_3\bigg) \bigg(U_2+ \big(1-\dfrac{p_2-2}{T_n}\big) U_3\bigg)^T\bigg] \nonumber\\
& =  \mathbb{E}\bigg[U_2 U_2^T\bigg] + \mathbb{E}\bigg[U_2  \big(1-\dfrac{p_2-2}{T_n}\big) U_3^T\bigg]+ \mathbb{E}\bigg[U_2^T  \{1-\dfrac{p_2-2}{T_n}\} U_3\bigg] + \mathbb{E}\bigg[\big(1-\dfrac{p_2-2}{T_n}\big)^2 U_3U_3^T\bigg]\nonumber\\
& = cov(U_2)+  \bigg[\mathbb{E}(U_2)\bigg]\bigg[\mathbb{E}(U_2)\bigg]^T + 2 \underbrace{\mathbb{E}\bigg[U_2  \big(1-\dfrac{p_2-2}{T_n}\big) U_3^T\bigg]}_{W_1} + \underbrace{\mathbb{E}\bigg[\big(1-\dfrac{p_2-2}{T_n}\big)^2 U_3U_3^T\bigg]}_{W_2}
\end{align}
Due to Lemmas, we can write:
\begin{align}\label{m1}
W_1 & =  \mathbb{E}\bigg[U_2  \big(1-\dfrac{p_2-2}{T_n}\big) U_3^T\bigg] \nonumber\\
& =\mathbb{E}\bigg[ \mathbb{E}\bigg[U_2  \big(1-\dfrac{p_2-2}{T_n}\big) U_3^T| U_3\bigg]\bigg] \nonumber\\
& =\mathbb{E}\bigg[ \mathbb{E}\bigg[U_2 | U_3\bigg] \big(1-\dfrac{p_2-2}{T_n}\big) U_3^T\bigg] \nonumber\\
& =\mathbb{E}\bigg[ \bigg\{ \mathbb{E}(U_2)+ cov(U_2, U_3) [cov(U_3)]^{-1} \big( U_3-\mathbb{E}(U_2) \big)\bigg\} \big(1-\dfrac{p_2-2}{T_n}\big) U_3^T\bigg] \nonumber\\
& = \mathbb{E}\bigg[-\boldsymbol{J\vartheta} \big(1-\dfrac{p_2-2}{T_n}\big) U_3^T\bigg] \nonumber\\
& =  -\boldsymbol{J\vartheta} \mathbb{E}\bigg[\big(1-\dfrac{p_2-2}{T_n}\big) U_3^T\bigg] \nonumber\\
& =  -\boldsymbol{J\vartheta} \mathbb{E} (U_3^T) \mathbb{E} \bigg[1- (p_2-2)  \chi^{-2}_{p_2+2} (\lambda)\bigg] \nonumber\\
& =  -\boldsymbol{W} \mathbb{E} \bigg[1- (p_2-2) \chi^{-2}_{p_2+2} (\lambda)\bigg] \nonumber\\
& =  -\boldsymbol{W}+ (p_2-2)\boldsymbol{W} \mathbb{E} \bigg[\chi^{-2}_{p_2+2} (\lambda)\bigg]
\end{align}
The same way for $W_2$, we have:
 \begin{align} \label{m2}
W_2 & = \mathbb{E}\bigg[\big(1-\dfrac{p_2-2}{T_n}\big)^2 U_3U_3^T\bigg] \nonumber\\
& = cov(U_3) \mathbb{E}\bigg[\big\{1-(p_2-2) \chi^{-2}_{p_2+2} (\lambda) \big\}^2\bigg]+ \bigg[\mathbb{E}(U_3)\bigg] \bigg[\mathbb{E}(U_3)\bigg]^T \mathbb{E}\bigg[\big\{1-(p_2-2) \chi^{-2}_{p_2+4} (\lambda) \big\}^2\bigg] \nonumber\\
& = \boldsymbol{J}\boldsymbol{F}^{-1} \mathbb{E}\bigg[\big\{1- (p_2-2)\chi^{-2}_{p_2+2} (\lambda)\big\}^2\bigg]+ \boldsymbol{W} \mathbb{E}\bigg[ \{1-(p_2-2)\chi^{-2}_{p_2+4} (\lambda) \}^2\bigg]
\end{align}
Thus,
\begin{align}
ADMSE(\hat{\boldsymbol{\theta}}_{JSZIB}) &=  \boldsymbol{F}^{-1}- \boldsymbol{J} \boldsymbol{F}^{-1} +  \boldsymbol{W}
 -2\boldsymbol{W} +2 (p_2-2) \boldsymbol{W} \mathbb{E} \bigg[\chi^{-2}_{p_2+2} (\lambda)\bigg] \nonumber \\
& \qquad + \boldsymbol{J}\boldsymbol{F}^{-1} \mathbb{E}\bigg[ \{1-(p_2-2)\chi^{-2}_{p_2+2} (\lambda) \}^2\bigg]+ \boldsymbol{W} \mathbb{E}\bigg[ \{1-(p_2-2) \chi^{-2}_{p_2+4} (\lambda) \}^2\bigg] \nonumber \\
&= \boldsymbol{F}^{-1}+ 2 (p_2-2) \boldsymbol{W} \bigg\{ \mathbb{E}\bigg[ \chi^{-2}_{p_2+2} (\lambda)\bigg]-\mathbb{E}\bigg[ \chi^{-2}_{p_2+4} (\lambda) \bigg] +\dfrac{p_2-2}{2}\mathbb{E}\bigg[ \chi^{-4}_{p_2+4} (\lambda) \bigg] \bigg\}\nonumber\\
& \qquad -2 (p_2-2)\boldsymbol{J}\boldsymbol{F}^{-1} \bigg\{ \mathbb{E}\bigg[ \chi^{-2}_{p_2+4} (\lambda) \bigg] -\dfrac{p_2-2}{2}\mathbb{E}\bigg[\chi^{-4}_{p_2+2} (\lambda)\bigg]\bigg\}
\end{align}
For the ADMSE of PJSZIB, we have:
\begin{align}\label{e4}
ADMSE(&\hat{\boldsymbol{\theta}}_{PJSZIB}) =\mathbb{E}\bigg[ \sqrt{n}(\hat{\boldsymbol{\theta}}_{PJSZIB}-\boldsymbol{\theta})\sqrt{n}(\hat{\boldsymbol{\theta}}_{PJSZIB}-\boldsymbol{\theta})^T \bigg] \nonumber\\
& = \mathbb{E}\bigg[ \sqrt{n}(\hat{\boldsymbol{\theta}}_{JSZIB}- \boldsymbol{\theta})\sqrt{n}(\hat{\boldsymbol{\theta}}_{JSZIB}- \boldsymbol{\theta})^T\bigg] - \mathbb{E}\bigg[\sqrt{n}(\hat{\boldsymbol{\theta}}_{JSZIB}- \boldsymbol{\theta})\big(1-\dfrac{p_2-2}{T_n}\big) U_3^T I_{(T_n<p_2-2)}\bigg]  \nonumber\\
& \qquad - \mathbb{E}\bigg[\sqrt{n}(\hat{\boldsymbol{\theta}}_{JSZIB}- \boldsymbol{\theta})^T \big(1-\dfrac{p_2-2}{T_n}\big)U_3 I_{(T_n<p_2-2)}\bigg]+ \mathbb{E}\bigg[\big(1-\dfrac{p_2-2}{T_n}\big)^2U_3 U_3^T  I_{(T_n<p_2-2)}\bigg] \nonumber\\
& = ADMSE(\hat{\boldsymbol{\theta}}_{JSZIB}) - 2 \underbrace{ \mathbb{E}\bigg[\sqrt{n}(\hat{\boldsymbol{\theta}}_{JSZIB}- \boldsymbol{\theta}) \big(1-\dfrac{p_2-2}{T_n}\big) U_3^T I_{(T_n<p_2-2)}\bigg]}_{W_3} \nonumber\\
& \qquad +\underbrace{\mathbb{E}\bigg[\big(1-\dfrac{p_2-2}{T_n}\big)^2U_3 U_3^T  I_{(T_n< p_2-2)}\bigg]}_{W_4}
\end{align}
where
\begin{align}\label{m3}
W_3 &=  \mathbb{E}\bigg[\sqrt{n}(\hat{\boldsymbol{\theta}}_{JSZIB}- \boldsymbol{\theta}) \big(1-\dfrac{p_2-2}{T_n}\big) U_3^T I_{(T_n< p_2-2)}\bigg] \nonumber \\
& = \mathbb{E}\bigg[\bigg(U_2+ \big(1- \dfrac{p_2-2}{T_n}\big) U_3\bigg) U_3^T\big(1- \dfrac{p_2-2}{T_n}\big) I_{(T_n<p_2-2)}\bigg] \nonumber\\
& =\mathbb{E}\bigg[U_2 U_3^T\big(1- \dfrac{p_2-2}{T_n}\big) I_{(T_n< p_2-2)}\bigg] + \mathbb{E}\bigg[ U_3 U_3^T\big(1- \dfrac{p_2-2}{T_n}\big)^2 I_{(T_n< p_2-2)}\bigg]\nonumber\\
& = \mathbb{E}(U_2)\mathbb{E}\bigg[ U_3^T  I_{\big(\chi^{2}_{p_2+2} (\lambda)< p_2-2\big)}\bigg] \nonumber \\
& \qquad+ cov(U_2, U_3)[cov(U_3)]^{-1}\mathbb{E}\bigg[ (U_3-\mathbb{E}(U_3))\big(1- \dfrac{p_2-2}{T_n}\big)  I_{\big(\chi^{2}_{p_2+2} (\lambda)< p_2-2\big)}\bigg]\nonumber \\
& \qquad + cov(U_3)\mathbb{E}\bigg[\bigg(1- (p_2-2)\chi^{-2}_{p_2+2} (\lambda)\bigg)^2 I_{\big(\chi^{2}_{p_2+2} (\lambda)< p_2-2\big)}\bigg]\nonumber \\
& \qquad + \mathbb{E}(U_3) \mathbb{E}(U_3^T) \mathbb{E}\bigg[\bigg(1- (p_2-2)\chi^{-2}_{p_2+4} (\lambda)\bigg)^2 I_{\big(\chi^{2}_{p_2+4} (\lambda)< p_2-2\big)}\bigg]\nonumber\\
&= \boldsymbol{W} \Phi_{p_2+2}(T_{(p_2, \alpha)}; \lambda)+ \boldsymbol{J}\boldsymbol{F}^{-1}\mathbb{E}\bigg[\bigg(1- (p_2-2)\chi^{-2}_{p_2+2} (\lambda)\bigg)^2 I_{\big(\chi^{2}_{p_2+2} (\lambda)< p_2-2\big)}\bigg]\nonumber \\
& \qquad + \boldsymbol{W} \mathbb{E}\bigg[\bigg(1- (p_2-2)\chi^{-2}_{p_2+4} (\lambda)\bigg)^2 I_{\big(\chi^{2}_{p_2+4} (\lambda)< p_2-2\big)}\bigg]
\end{align}
and
\begin{align} \label{m4}
W_4 &=  \mathbb{E}\bigg[\big(1-\dfrac{p_2-2}{T_n}\big)^2U_3 U_3^T  I_{(T_n< c)}\bigg]\nonumber \\
&= cov(U_3) \mathbb{E}\bigg[\bigg(1- (p_2-2)\chi^{-2}_{p_2+4} (\lambda)\bigg)^2  I_{\big(\chi^{2}_{p_2+4} (\lambda)< p_2-2\big)} \bigg]\nonumber \\
& \qquad + \mathbb{E}(U_3) \mathbb{E}(U_3^T) \mathbb{E}\bigg[\bigg(1- (p_2-2)\chi^{-2}_{p_2+2} (\lambda)\bigg)^2  I_{\big(\chi^{2}_{p_2+2} (\lambda)< p_2-2\big)} \bigg]\nonumber \\
&=\boldsymbol{J}\boldsymbol{F}^{-1} \mathbb{E}\bigg[\bigg(1- (p_2-2)\chi^{-2}_{p_2+2} (\lambda)\bigg)^2  I_{\big(\chi^{2}_{p_2+2} (\lambda)< p_2-2\big)} \bigg] + \boldsymbol{W}\mathbb{E}\bigg[\bigg(1- (p_2-2)\chi^{-2}_{p_2+4} (\lambda)\bigg)^2  I_{\big(\chi^{2}_{p_2+4} (\lambda)< p_2-2\big)} \bigg]
\end{align}
By replacing \eqref{m3} and \eqref{m4} in \eqref{e4}, we will have:
\begin{align}\label{e5}
ADMSE(\hat{\boldsymbol{\theta}}_{PJSZIB}) &=ADMSE(\hat{\boldsymbol{\theta}}_{JSZIB}) - 2 \boldsymbol{W} \Phi_{p_2+2}(T_{(p_2, \alpha)}; \lambda)\nonumber \\
& \qquad-2\boldsymbol{J}\boldsymbol{F}^{-1}\mathbb{E}\bigg[\bigg(1- (p_2-2)\chi^{-2}_{q+2} (\lambda)\bigg)^2 I_{\big(\chi^{2}_{p_2+2} (\lambda)< p_2-2\big)}\bigg]\nonumber \\
& \qquad -2 \boldsymbol{W} \mathbb{E}\bigg[\bigg(1- (p_2-2)\chi^{-2}_{p_2+4} (\lambda)\bigg)^2 I_{\big(\chi^{2}_{p_2+4} (\lambda)< p_2-2\big)}\bigg]\nonumber \\
& \qquad +\boldsymbol{J}\boldsymbol{F}^{-1} \mathbb{E}\bigg[\bigg(1- (p_2-2)\chi^{-2}_{p_2+2} (\lambda)\bigg)^2  I_{\big(\chi^{2}_{p_2+2} (\lambda)< p_2-2\big)} \bigg]\nonumber \\
& \qquad + \boldsymbol{W}\mathbb{E}\bigg[\bigg(1- (p_2-2)\chi^{-2}_{p_2+4} (\lambda)\bigg)^2  I_{\big(\chi^{2}_{p_2+4} (\lambda)< p_2-2\big)} \bigg]\nonumber\\
&=ADMSE(\hat{\boldsymbol{\theta}}_{JSZIB}) - 2 \boldsymbol{W} \Phi_{p_2+2}(T_{(p_2, \alpha)}; \lambda)\nonumber \\
& \qquad-\boldsymbol{J}\boldsymbol{F}^{-1}\mathbb{E}\bigg[\bigg(1- (p_2-2)\chi^{-2}_{p_2+2} (\lambda)\bigg)^2 I_{\big(\chi^{2}_{p_2+2} (\lambda)< p_2-2\big)}\bigg]\nonumber \\
& \qquad - \boldsymbol{W} \mathbb{E}\bigg[\bigg(1- (p_2-2)\chi^{-2}_{p_2+4} (\lambda)\bigg)^2 I_{\big(\chi^{2}_{p_2+4} (\lambda)< p_2-2\big)}\bigg]
\end{align}
For $\hat{\boldsymbol{\theta}}_{PTZIB}$, we have:
\begin{align}
ADMSE(\hat{\boldsymbol{\theta}}_{PJSZIB}) &=\mathbb{E}\bigg[ \sqrt{n}(\hat{\boldsymbol{\theta}}_{PJSZIB}^{PTE}-\boldsymbol{\theta})\sqrt{n}(\hat{\boldsymbol{\theta}}_{PJSZIB}-\boldsymbol{\theta})^T \bigg] \nonumber\\
& = \mathbb{E}\bigg[\bigg(U_1- U_3 I_{\big(T_n \leq T_{(p_2,\alpha)}\big)}\bigg)\bigg(U_1- U_3 I_{\big(T_n \leq T_{(p_2,\alpha)}\big)}\bigg)^T \bigg] \nonumber\\
& = \mathbb{E}\bigg[U_1U_1^T\bigg]- \mathbb{E}\bigg[U_1U_3^T I_{\big(T_n \leq T_{(p_2,\alpha)}\big)}\bigg]- \mathbb{E}\bigg[U_1^TU_3 I_{\big(T_n \leq T_{(p_2,\alpha)}\big)}\bigg] + \mathbb{E}\bigg[U_3U_3^T I_{\big(T_n \leq T_{(p_2,\alpha)}\big)}\bigg]  \nonumber\\
& = cov(U_1)+ \mathbb{E}(U_1)\mathbb{E}(U_1^T) -2 \mathbb{E}\bigg[U_1^TU_3 I_{\big(T_n \leq T_{(p_2,\alpha)}\big)}\bigg]+
\mathbb{E}\bigg[U_3U_3^T I_{\big(T_n \leq T_{(p_2,\alpha)}\big)}\bigg]  \nonumber\\
& = \boldsymbol{F}^{-1} -2 \mathbb{E}\bigg[\mathbb{E}\bigg[U_1U_3^T I_{\big(T_n \leq T_{(p_2,\alpha)}\big)}|U_3 \bigg] \bigg]+  \mathbb{E}\bigg[U_3U_3^T I_{\big(T_n \leq T_{(p_2,\alpha})\big)}\bigg]  \nonumber\\
& =  \boldsymbol{F}^{-1} -2\mathbb{E}(U_1) \mathbb{E}\bigg[U_3^T I_{\big(T_n \leq T_{(p_2,\alpha)}\big)}\bigg] -2 cov(U1, U_3)[cov(U_3)]^{-1} \mathbb{E}\bigg[(U_3 -\mathbb{E}(U_3)) I_{\big(T_n \leq T_{(p_2,\alpha)}\big)}\bigg] \nonumber\\
& \qquad+  cov(U_3) \Phi_{p_2+2}(T_{(p_2, \alpha)}; \lambda)+\mathbb{E}(U_3)\mathbb{E}(U_3^T)  \Phi_{p_2+4}(T_{(p_4, \alpha)}; \lambda) \nonumber\\
& =  \boldsymbol{F}^{-1} + \boldsymbol{J}\boldsymbol{F}^{-1}\Phi_{p_2+2}(T_{(p_2, \alpha)}; \lambda)+\boldsymbol{W}\Phi_{p_2+4}(T_{(p_4, \alpha)}; \lambda) \nonumber
\end{align}
\end{appendices}


\begin{thebibliography}{99}

\bibitem{Abduljabbar}
Abduljabbar, L.A., \& Z.Y. Algamal (2022) {\em Jackknifed K-L estimator in Bell regression model}, Mathematics Theory and its Contribution in Robotics and Computer Engineering, \textbf{71}, 267--278.

\bibitem{Ahmed}
Ahmed, S.E. (2014) {\em Penalty, shrinkage and pretest strategies: variable selection and estimation}. Springer, New York.

\bibitem{Akram}
Akram, M.N., M. Amin, \& M. Qasim (2023) {\em A new biased estimator for the gamma regression model: Some applications in medical sciences}, Communications in Statistics - Theory and Methods, \textbf{52(11)}, 3612--3632. 

\bibitem{Akram2}
Akram, M.N.,  M. Amin, \& M. Amanullah (2021) {\em James Stein estimator for the inverse Gaussian regression model}. Iran J Sci Technol Trans Sci,\textbf{45} , 1389--1403.

\bibitem{Algamal1}
Algamal, Z., A. Lukman, B.M.K. Golam, \& A. Taofik (2023) {\em Modified Jackknifed Ridge Estimator in Bell Regression Model: Theory, Simulation and Applications}, Iraqi Journal For Computer Science and Mathematics, \textbf{4(1)},  146--154.

\bibitem{Algamal2}
Algamal, Z.Y., A.F. Lukman, M.R. Abonazel, \& F.A. Awwad (2022) {\em  Performance of the Ridge and Liu Estimators in the zero-inflated Bell Regression Model}, Journal of Mathematics, ID 9503460. 

\bibitem{Ali}
Ali, E., M.L. Diop, \& A.  Diop (2022) {\em Statistical Inference in a Zero-Inflated Bell Regression Model}, Math. Meth. Stat, \textbf{31}, 91--104. 

\bibitem{Amin}
Amin, M., M.N. Akram, \& A. Majid (2023) {\em  On the estimation of Bell regression model using ridge estimator}, Communications in Statistics - Simulation and Computation, \textbf{52(3)}, 854--867. 

\bibitem{Amin2}
Amin, M., M.N. Akram, \& M. Amanullah (2022) {\em  On the James-Stein estimator for the Poisson regression model}, Communications in Statistics - Simulation and Computation, 51:10, 5596--5608. 

\bibitem{Bancroft}
Bancroft, T.A. (1944) {\em On biases in estimation due to the use of preliminary tests of significance}. Ann Math Stat,\textbf{15}, 190--204.

\bibitem{Bell1}
Bell, E.T. (1934) {\em Exponential numbers}, The American Mathematical Monthly,\textbf{41}, 419.

\bibitem{Bell2}
Bell, E.T. (1934) {\em Exponential polynomials}, Annals of Mathematics, \textbf{35}, 258.

\bibitem{Castellares}
Castellares,  F., S.L.P. Ferrari, \& A. J. Lemonte (2017) {\em On the Bell distribution and its associated regression model for count data}, Applied Mathematical Modelling, \textbf{56}. 

\bibitem{Davidson}
Davidson, R.R., \& W.E. Lever, (1970){\em The limiting distribution of the likelihood ratio statistic under a class of local alternatives}, Sankhya A, \textbf{32(2)}, 209--224.

\bibitem{bellreg}
Demarqui, F., M. Prates, \& F. Caceres (2022) {\em bellreg  R Package}. https://cran.r-project.org/web/packages/bellreg/index.html.

\bibitem{Ertan}
 Erko\c{c}, A., \& K.U. Akay (2023) {\em A new improvement Liu-type estimator for the Bell regression model}, Communications in Statistics - Simulation and Computation. DOI: 10.1080/03610918.2023.2252624.

\bibitem{Heyde}
Heyde, C.C. (2008) {\em Quasi-Likelihood and its Application: A General Approach to Optimal Parameter Estimation}, Springer Science and Business Media.

\bibitem{judge}
Judge, G.G., \&  M.E. Bock (1978) {\em The statistical implications of pre-test and stein-rule estimators in econometrics}, North Holland Publishing Company, Amsterdam.

\harvarditem{Kibria and Saleh}{2004}{Kibria}
Kibria, B.M.G., \& A. Saleh (2004) {\em Performance of positive rule estimator in the ill-conditioned Gaussian regression model}. Calcutta Statist Assoc Bull, \textbf{55} , 209--239.

\bibitem{Lemonte}
Lemonte, A.J., G. Moreno-Arenas, \& F. Castellares (2020) {\em Zero-inflated Bell regression models for count data}, Journal of Applied Statistics, \textbf{47(2)}, 265-286. 

\bibitem{Majid}
Majid, A., M. Amin, \& M.N. Akram (2022) {\em On the Liu estimation of Bell regression model in the presence of multicollinearity}, Journal of Statistical Computation and Simulation,\textbf{92(2)} , 262--282. 

\bibitem{Nelder}
 Nelder, J.A., \& R.W. Wedderburn (1972) {\em Generalized linear models}. Journal of the Royal Statistical Society: Series A (General), \textbf{135(3)}, 370--384.

\bibitem{Seifollahi1}
Seifollahi, S., \& H. Bevrani (2023) {\em  James-Stein type estimators in beta regression model: simulation and application}. Hacettepe Journal of Mathematics and Statistics, 52(4), 1046–1065, 2023. 

\bibitem{Seifollahi2}
Seifollahi, S., H. Bevrani, \& Z.Y. Algamal (2023) {\em  Improved estimators in Bell regression model with application}. https://doi.org/10.48550/arXiv.2401.00966.

\bibitem{Saffari}
Saffari, S.E., \& R. Adnan (2012) {\em Parameter estimation on zero-inflated negative binomial regression with right truncated data}, Sains Malaysiana, \textbf{41}, 1483-1487.

\bibitem{Stein}
 Stein, C. (1956) {\em The admissibility of hotelling’s T 2-test}. Math, Stat., \textbf{27}, 616--623.
 
\bibitem{Zandi}
Zandi, Z., H. Bevrani, \& R. Arabi Belaghi (2021) {\em Estimation of fixed parameters in negative binomial mixed model using shrinkage estimators}, Journal of Data Science and Modeling, \textbf{1(2)}, 99-124.


\end{thebibliography}
\end{document}